\documentclass[a4paper,10pt]{article}

\usepackage[english]{babel}
\usepackage[ansinew]{inputenc}
\usepackage{amsmath}
\usepackage{amssymb}
\usepackage{amsthm}
\usepackage{stmaryrd}
\usepackage{mathrsfs}
\usepackage{enumerate}
\usepackage{dsfont}
\usepackage{amsthm}
\usepackage{geometry}
\usepackage{accents}
\usepackage{array}
\usepackage{cite}
\usepackage{color}
\usepackage{hyperref}

\usepackage{fancyhdr}
\definecolor{mygray}{gray}{.3}

\numberwithin{equation}{section}

\newcommand{\mH}{\mathcal{H}}
\newcommand{\tW}{\tilde{W}}
\newcommand{\tL}{\tilde{\Lambda}}
\newcommand{\up}{\underline{p}}
\newcommand{\beq}{\begin{equation*}}
\newcommand{\eeq}{\end{equation*}}
\newcommand{\bneq}{\begin{equation}}
\newcommand{\eneq}{\end{equation}}

\newtheorem{lemma}{Lemma}[section]
\newtheorem{theorem}{Proposition}[section]

\author{Matthias Plaschke \\ \\ \small{Faculty of Physics, University of Vienna}}
\title{Wedge Local Deformations of Charged Fields leading to Anyonic Commutation Relations}
\date{\today}

\begin{document}
\maketitle

\begin{abstract}
The method of deforming free fields by using multiplication operators on Fock space, introduced in \cite{Lechner12}, is generalized to a \emph{charged} free field on two- and three-dimensional Minkowski space. In this case the deformation function can be chosen in such a way that the deformed fields satisfy generalized commutation relations, i.e. they behave like Anyons instead of Bosons. \\
The fields are ``polarization free'' in the sense that they create only one-particle states from the vacuum and they are localized in wedges (or ``paths of wedges''), which makes it possible to circumvent a No-Go theorem by J. Mund \cite{Mund98}, stating that there are no free Anyons localized in spacelike cones. The two-particle scattering matrix, however, can be defined and is different from unity.
\end{abstract}

\tableofcontents

\hspace{1em}

\section{Introduction}
In \cite{Lechner12} it has been shown how a free hermitian scalar Bose field can be deformed such that the S-matrix of the model becomes non-trivial. This can be regarded as a generalization of earlier work by Grosse and Lechner \cite{GrossLech, GrossLech07} where they show how a quantum field on Moyal-Minkowski spacetime can be understood as a deformation of a quantum field on ordinary spacetime. In 1+1 dimensions this leads to a family of integrable models, namely those with a factorizing S-matrix, satisfying $S(0) = 1$. By deforming free fields on a \emph{fermionic} Fock space, this procedure has been generalized in \cite{Alazzawi12} to include also models with $S(0) = -1$. On Fock space such deformations are defined using multiplication operators $(T_R(Qp)\Psi)(p_1,...,p_n) \sim \prod R(Qp\cdot p_i) \Psi(p_1,...,p_n)$ where $Q$ is an antisymmetric ``deformation matrix'' depending on a wedge $W$ (see also \cite{GrossLech}). A wedge region $W$ is defined to be the Lorentz transform of the standard wedge $W_0 = \{x\in \mathds{R}^d : x_1 > |x_0|\}$ and its causal complement will be denoted by $W'$. These operators are then used to deform the creation and annihilation operators of the free field $\Phi(f) = a^*(f^+) + a(\bar{f}^+)$ according to $a_{R,Q}(p) := a(p) T_R(Qp)$ etc. Here $f \in \mathscr{S}$ is a testfunction and $f^+$ denotes the restriction of its Fourier transform to the upper mass shell. One can then show that under certain conditions on the functions $R$ the deformed fields $\Phi_{R,Q}(f) := a^*_{R,Q}(f^+) + a_{R,Q}(\bar{f}^+)$ are still Poincaré covariant Wightman fields which are no longer localizable in compact regions but they are localized in wedges. In two dimensions, however, it is possible to show that the algebras for double cones \footnote{In $d=1+1$ double cones are obtained by intersection of two opposite translated wedges.} still satisfy the Reeh-Schlieder property if the S-matrix of the model fulfills some kind of additional regularity condition \cite{Lechner08}. \\
In this work we want to generalize this procedure to a \emph{charged} scalar field in low dimensions and we will see that in this case it is also possible to change the statistics of the fields in such a way that they satisfy anyonic commutation relations. This is an interesting result because in their recent paper \cite{BrosMund} Bros and Mund show that if a quantum field theory has tempered polarization-free generators the corresponding fields necessarily satisfy Bose- or Fermi commutation relations. This is a generalization of the original No-Go theorem by Mund \cite{Mund98} and the reason why we can circumvent it is that in our case the algebras for regions smaller than a wedge are ``too small'' in the sense that they don't generate a dense set in the Hilbert space when acting on the vacuum (Reeh-Schlieder property).\\ 
To construct such a deformation we consider as Hilbert space the symmetric Fock space $\mH = \mathcal{F}_s(\mH_1)$ over the doubled one-particle space $\mH_1 := \mH_1^+\oplus\mH_1^-$, where $\mH_1^\pm = L^2(\mathds{R},d\theta)$ in the two-dimensional case and $\mH_1^\pm = L^2(\mathds{R}^3,d\mu)$ in the three dimensional case. \sloppy{Here \mbox{$d\mu(p) = \delta(p^2-m^2)\theta(p_0) d^3p$}} denotes the Lorentz invariant measure on the mass shell and $\theta$ denotes the rapidity, which is defined by $p(\theta) = m \binom{\cosh(\theta)}{\sinh(\theta)}$. The Hilbert space $\mH$ is isomorphic to the tensor product $\mathcal{F}_s(\mH_1^+)\otimes\mathcal{F}_s(\mH_1^-)$ and we will refer to the first tensor factor as the ``particle space'' and to the second as the ``anti-particle space''. Because of this tensor product structure we will consider vectors of the form $\Psi_n\otimes\Psi_m$ which we denote by $\Psi_n^m$ for simplicity. On this doubled Fock space we now have a charge conjugation operator $C$, which simply exchanges the two factors, and a charge operator $(Q\Psi)_n^m(p_1,...,p_{n+m}) := (n-m)\Psi_n^m(p_1,...,p_{n+m})$. We also naturally have two sets of creation and annihilation operators, $a, a^*$ and $b, b^*$, which are defined according to
\beq
\begin{split}
(a(\varphi) \Psi)_n^m(p_1,...,p_{n+m}) &= \sqrt{n+1} \int d\mu(p) \overline{\varphi(p)} \Psi_{n+1}^m(p,p_1,...,p_{n+m}) \ , \hspace{4em} a^*(\varphi) := a(\varphi)^* \\
(b(\varphi) \Psi)_n^m(p_1,...,p_{n+m}) &= \sqrt{m+1} \int d\mu(p) \overline{\varphi(p)} \Psi_n^{m+1}(p_1,...,p_n,p,p_{n+1},...,p_{n+m}) \ , \hspace{1em} b^*(\varphi) := b(\varphi)^*
\end{split}
\eeq
and their distributional kernels satisfy the canonical commutation relations \footnote{$a^\sharp$ stands for either $a^*$ or $a$.}
\beq
\begin{split}
[a^\sharp(p),a^\sharp(p')] &= 0 \ , \hspace{5.5em} \ [a^\sharp(p),b^\sharp(p')] = 0, \\
[a(p),a^*(p')] &= \omega_p\ \delta(p-p')\ , \hspace{1em} \ [a(p),b^*(p')] = 0,
\end{split}
\eeq
with $\omega_p = \sqrt{\mathbf{p}^2+m^2}$. Also the corresponding charge conjugated relations hold, where $a$ and $b$ are interchanged. Using these creation and annihilation operators one can then define the free field $\Phi(f) := a^*(f^+) + b(\bar{f}^+)$ and the charge conjugate field $\Phi^*(f) = b^*(f^+) + a(\bar{f}^+)$. \\ 
Additionally we will have a representation $U(a,\Lambda)$ of the Poincaré group together with a space-time reflection $J$, which acts according to
\bneq \label{jactionbeta}
(J\Psi)_n^m(p_1,..,p_{n+m}) := e^{i\beta q}\, \overline{\Psi_n^m(-j p_1,...,-jp_{n+m})},
\eneq
where $j$ is the reflection at the $x_2$ axis, $j(x_0,x_1,x_2) = (-x_0,-x_1,x_2)$, and $\beta$ is a parameter which will be specified below. In the rapidity parametrization in $d=1+1$ this simplifies to $(J\Psi)_n^m(\theta_1,...,\theta_{n+m}) = e^{i\beta q}\, \overline{\Psi_n^m(\theta_1,...,\theta_{n+m})}$.
\\

\section{Deformations on Two-Dimensional Minkowski Space}
\subsection{No-Go Theorem for Compactly Localized ``Free'' Anyons}
On the charged Hilbert space we now want to construct covariant quantum fields with anyonic commutation relations as deformations of ordinary (free) Bose fields on the bosonic Fock-space. The simplest possibility would be fields $\Phi_\lambda$ that create one-particle states when applied to the vacuum $\Omega$, and fulfill commutation relations of the form
\bneq  \label{CommrelComp}
\begin{split}
\Phi_\lambda(f)\Phi^*_\lambda(g) &= e^{2\pi i \lambda \epsilon(f,g)} \Phi^*_\lambda(g)\Phi_\lambda(f), \\
\Phi_\lambda(f)\Phi_\lambda(g) &= e^{-2\pi i \lambda \epsilon(f,g)} \Phi_\lambda(g)\Phi_\lambda(f)
\end{split}
\eneq
whenever $f$ and $g$ have spacelike separated support. The factor $\epsilon(f,g)$ in the exponential measures if the testfunction $g$ is supported to the right or the left spacelike complement of $\operatorname{supp}f$ (which is a well-defined Poincaré-invariant concept in $d=1+1$), such that $\epsilon(f,g) = 1$ if $g$ lies to the left of $f$ and $\epsilon(g,f) = -\epsilon(f,g)$. \\
Unfortunately such fields are not possible if one assumes (apart from the fields being one-particle generators) that they satisfy the Reeh-Schlieder property and makes the technical assumption that the fields are temperate. This means that for any two spacelike separated testfunctions $f$ and $g$ the vector $U(x)\Phi_\lambda(g)\Omega$ (where $U(x)$ denotes the representation of the translations) is in the domain of $\Phi_\lambda(f)$ for all $x \in \mathds{R}^2$ and the function
\beq
x \mapsto ||\Phi_\lambda(f)U(x)\Phi_\lambda(g)\Omega||
\eeq
is locally integrable and polynomially bounded for large $x$.
This is made explicit in the following theorem.
\begin{theorem}
Consider a Wightman field $\Phi_\lambda$ satisfying the commutation relations \eqref{CommrelComp}. If $\Phi_\lambda$ is temperate, satisfies the Reeh-Schlieder property and generates one-particle vectors from the vacuum, then $\lambda \in \frac{\mathds{N}}{2}$, i.e. $\Phi_\lambda$ satisfies Bose- or Fermi statistics.
\end{theorem}
\begin{proof}
The impossibility to construct fields satisfying the above requirements follows from the same arguments as in the No-Go theorem for string-localized free Anyons in $d=2+1$ in \cite{Mund98}. First one can show that only a multiple of $\Omega$ is added to the commutation relations \eqref{CommrelComp} if one translates the localization regions, such that they are not spacelike separated any more, i.e.
\bneq \label{JostSchroer}
\begin{split}
&U(x)\Phi_\lambda(f)U(-x)\Phi^*_\lambda(g)\Omega - e^{2\pi i \lambda \epsilon(f,g)}\Phi^*_\lambda(g)U(x)\Phi_\lambda(f)\Omega  \\
\equiv &\Phi_\lambda(\alpha_x(f))\Phi^*_\lambda(g)\Omega - \omega^{\epsilon(f,g)}\Phi^*_\lambda(g)\Phi_\lambda(\alpha_x(f))\Omega = C_{f,g}(x)\Omega
\end{split}
\eneq
To prove this relation one uses the aforementioned temperateness condition to show that for two spacelike separated fields $\phi_1$ and $\phi_2$ the $\mH$-valued function $U(x)\phi_1U(-x)U(y)\phi_2U(-y)\Omega$ is a tempered distribution, whose Fourier transform has support contained in $(H_m^-\cup H_m^+)\times H_m^+$, where $H_m^\pm$ denotes the upper/lower mass shell. Defining $F^+$ and $F^-$ by
\beq
U(x)\phi_1 U(-x)U(y)\phi_2 U(-y)\Omega =: F_{1,2}^+(x,y)+F_{1,2}^-(x,y) \ , \hspace{1em} \ \mathrm{with}\ \ supp\widetilde{F^\pm} \subset H_m^\pm\times H_m^+,
\eeq
one gets that
\beq
\begin{split}
F_{1,2}^-(x,y) &= \langle\Omega,U(x)\phi_1U(-x)U(y)\phi_2U(-y)\Omega\rangle\cdot \Omega \\
sp_P F_{1,2}^+(x,y) &\subseteq H_m^+ + H_m^+,
\end{split}
\eeq
where $sp_P \Psi$ denotes the spectral support of $\Psi$ w.r.t. the energy-momentum operators\footnote{i.e. the set of points $p\in spec(P)$ such that for any neighborhood $V$ of $p$, the spectral projector
$E_V(P)$ does not map $\Psi$ to zero.}. Using these spectral properties and the edge of the wedge theorem then leads to $F_{1,2}^+(x,0) - \omega F_{2,1}^+(0,x) = 0$ which finally leads to Eqn. \eqref{JostSchroer}. Moreover, the relation \eqref{JostSchroer} extends from $\Omega$ to the dense subspace $\mathcal{F}(\operatorname{supp} f)'\Omega\cap\mathcal{F}(\operatorname{supp} g)'\Omega\cap\mathcal{F}(\operatorname{supp} f + x)'\Omega$, where $\mathcal{F}(\mathcal{O})$ denotes the algebra generated by the fields localized in $\mathcal{O}$. (For a more detailed and mathematically rigorous formulation of these arguments see \cite{Mund98}.) \\
Next we chose two testfunctions $f$ and $g$ such that $f$ is localized to the left of $g$. Then we know because of relation \eqref{JostSchroer} that on a dense subspace there holds
\beq
\Phi_\lambda(\alpha_x f) \Phi_\lambda(g)^* = \omega\ \Phi_\lambda(g)^* \Phi_\lambda(\alpha_x f) + C_{f,g}(x)
\eeq
for all $x\in\mathds{R}^2$. On the other hand the commutation relations \eqref{CommrelComp} state that
\beq
\Phi_\lambda(\alpha_x f) \Phi_\lambda(g)^* = \omega^{-1} \Phi_\lambda(g)^* \Phi_\lambda(\alpha_x f) ,
\eeq
holds for all $x$ such that $\alpha_x f$ lies to the \emph{right} of $g$. Combining these two equations we see that
\beq
\Phi_\lambda(g)^* \Phi_\lambda(\alpha_x f) (\omega^{-1}-\omega) = C_{f,g}(x)
\eeq
for all $x$ where $\alpha_x f$ is localized to the right of $g$. Now suppose that $\omega^2 \neq 1$ (i.e. $\omega^{-1}-\omega\neq 0$). Then $\Phi_\lambda(g)^* \Phi_\lambda(\alpha_x f)$ is proportional to the identity when acting on the vacuum vector and thus proportional to the identity on its whole domain of definition because the vacuum is separable. Therefore $\Phi_\lambda(g)^* \Phi_\lambda(\alpha_x f)$ is invariant under Poincaré transformations and together with locality this implies that the individual fields $\Phi_\lambda(\alpha_x f), \Phi_\lambda(g)^*$ are already proportional to the identity, which is inconsistent with anyonic commutation relations.  This shows that $\omega^2 = 1$, which corresponds to Bose- or Fermi-statistics. 
\end{proof}
\vspace{1em}
The above proof rests upon the assumption that the fields are localizable in compact regions in Minkowski space, because only then it is possible to have three of them mutually spacelike separated. Therefore, we will now try to construct Anyon fields with weaker localization properties, namely localization in wedges. For this purpose we will use the recent construction by Lechner \cite{Lechner12} to obtain wedge-localized polarization-free generators as a deformation of free Bose fields.
\\

\subsection{Wedge-local Anyons}
Following \cite{Lechner12} we can define multiplication operators $T_R(\theta)$ on $\mathcal{F}_s(\mH_1^\pm)$ according to 
\bneq
(T_R(\theta)\Psi)_n(\theta_1,...,\theta_n) := \prod_{j=1}^n R(\theta-\theta_j) \Psi_n(\theta_1,...,\theta_n),
\eneq
with a ``deformation function'' $\theta \mapsto R(\theta)$ satisfying $\overline{R(\theta)} = R(\theta)^{-1}$ but \emph{not} necessarily $R(-\theta) = \overline{R(\theta)}$ (The condition $R(-\theta) = R(\theta)^{-1}$ would lead to $R(0) = \pm 1$ and therefore to Bose- or Fermi statistics).  This can now be generalized to the charged situation at hand and to anyonic statistics as follows. On $\mathcal{F}_s(\mH_1^+)\otimes\mathcal{F}_s(\mH_1^-)$ we define
\bneq
T_{R,r}(\theta) := e^{i\frac{\rho}{2}} (T_R(\theta)\otimes T_r(\theta)),
\eneq
with functions $R$ and $r$ satisfying $R(-\theta) = e^{i\mu} \overline{R(\theta)}$ and $r(-\theta) = e^{i\nu}\overline{r(\theta)}$ and three yet undefined parameters $\mu,\nu,\rho \in\mathds{R}$. Denoting the charge conjugation again by $C$ and taking $\beta = 0$ in the definition of the space-time reflection $J$ the operator $T_{R,r}$ transforms according to
\beq
C T_{R,r}(\theta) C = T_{r,R}(\theta) , \hspace{1em} \ J T_{R,r}(\theta) J = T_{R,r}(\theta)^* = e^{-i\rho} T_{\bar{R},\bar{r}}(\theta).
\eeq
This multiplication operator is now used to define deformed particle annihilation operators
\bneq
a_{R,r}(\theta) := a(\theta) T_{R,r}(\theta) = e^{i\frac{\rho}{2}} (a_R(\theta)\otimes T_r(\theta)),
\eneq
where $a_R(\theta)$ denotes the standard ``Lechner deformed'' operator on $\mathcal{F}_s(\mH_1^+)$.
The charge conjugated operator $b_{R,r}(\theta)$ then turns out to be
\bneq
b_{R,r}(\theta) := C a_{R,r}(\theta) C = b(\theta) T_{r,R}(\theta) = e^{i\frac{\rho}{2}}(T_r(\theta)\otimes b_R(\theta)).
\eneq
Using the transformation properties of $T_{R,r}(\theta)$ the space-time reflected operators are $Ja_{R,r}(\theta)J = e^{-i\rho}a_{\bar{R},\bar{r}}(\theta)$ and the adjoint operators are of course defined as $a^*_{R,r}(\theta) := a_{R,r}(\theta)^*$. \\
To determine the locality properties of the deformed fields we will need the commutation relations between the various creation and annihilation operators and between them and the operators $T_{R,r}(\theta)$. A straightforward calculation yields
\beq
\begin{split}
a(\theta)T_{R,r}(\theta ') = R(\theta '-\theta) T_{R,r}(\theta ')a(\theta) \ , \hspace{1em} \ &b(\theta)T_{R,r}(\theta ') = r(\theta '-\theta) T_{R,r}(\theta ')b(\theta), \\
a^*(\theta)T_{R,r}(\theta ') = R(\theta '-\theta)^{-1} T_{R,r}(\theta ')a^*(\theta) \ , \hspace{1em} \ &b^*(\theta)T_{R,r}(\theta ') = r(\theta '-\theta)^{-1} T_{R,r}(\theta ')b^*(\theta).
\end{split}
\eeq
This leads to the commutation relations
\bneq
\begin{split}
a^\sharp_{R,r}(\theta)a^\sharp_{\bar{R},\bar{r}}(\theta ') &= e^{-i\mu}a^\sharp_{\bar{R},\bar{r}}(\theta ')a^\sharp_{R,r}(\theta), \\
a^\sharp_{R,r}(\theta)b^\sharp_{\bar{R},\bar{r}}(\theta ') &= e^{-i\nu}b^\sharp_{\bar{R},\bar{r}}(\theta ')a^\sharp_{R,r}(\theta).
\end{split}
\eneq
The relations with $a$ and $b$ interchanged follow by charge conjugation and noting that $\mu$ and $\nu$ depend on $R$ and $r$ and satisfy $\mu(\bar{R}) = -\mu(R)$. The commutation relations for mixed creation and annihilation operators turn out to be
\bneq
\begin{split}
a_{R,r}(\theta)a^*_{\bar{R},\bar{r}}(\theta ') &= e^{i\mu}a^*_{\bar{R},\bar{r}}(\theta ')a_{R,r}(\theta) + e^{i(\mu-\rho)}\delta(\theta-\theta ')T_{R,r}(\theta)^2, \\
b_{R,r}(\theta)b^*_{\bar{R},\bar{r}}(\theta ') &= e^{i\mu}b^*_{\bar{R},\bar{r}}(\theta ')b_{R,r}(\theta) + e^{i(\mu-\rho)}\delta(\theta-\theta ')T_{r,R}(\theta)^2.
\end{split}
\eneq
Note that the last term is different for the $a$'s and $b$'s if $R\neq r$! Using the trivial commutation relations between $a$ and $b^*$ we also get
\bneq
a_{R,r}(\theta)b^*_{\bar{R},\bar{r}}(\theta ') = e^{i\nu} b^*_{\bar{R},\bar{r}}(\theta ')a_{R,r}(\theta),
\eneq
and the charge conjugated relation with $a$ and $b$ interchanged. \\
\vspace{1em}

Having computed all the necessary relations of the creation and annihilation operators we can now define the deformed fields. For a test function $f\in\mathscr{S}(\mathds{R}^2)$ we define
\bneq
\begin{split}
\Phi_{R,r}(f) &:= a^*_{R,r}(f^+) + b_{R,r}(\bar{f}^+), \\
\Phi^*_{R,r}(f) &= b^*_{R,r}(f^+) + a_{R,r}(\bar{f}^+) = C \Phi_{R,r}(f) C,
\end{split}
\eneq
with $f^\pm(\theta) := \frac{1}{2\pi}\int d^2x f(\pm x) e^{ip(\theta) x}$. We also need the field for the reflected wedge
\bneq
J \Phi_{R,r}(\alpha_j f) J = e^{i\rho}a^*_{\bar{R},\bar{r}}(f^+) +e^{-i\rho}b_{\bar{R},\bar{r}}(\bar{f}^+) \equiv \widehat{\Phi}_{R,r}(f),
\eneq
where the reflected test function $\alpha_j f$ is defined as $(\alpha_j f)(x) := \overline{f(-x)}$. \\
A straightforward computation, using the aforementioned relations for the deformed creation and annihilation operators, then shows that for the fields to satisfy simple commutation relations we need to set $\nu = -\mu$,\footnote{Note that if we would work with only one set of annihilation/creation operators on a single Fock space, the condition $e^{i \mu} = e^{-i \mu}$ would lead to $\mu = k\pi$ with $k\in 2\mathds{Z}+1$, i.e. to Bose- or Fermi statistics.} which then leads to
\bneq
\Phi_{R,r}(f) \widehat{\Phi}_{R,r}(g) = e^{-i\mu} \widehat{\Phi}_{R,r}(g)\Phi_{R,r}(f). 
\eneq
Next we want to calculate the commutation relations between $\Phi_{R,r}(f)$ and $\widehat{\Phi}^*_{R,r}(g)$ which leads to
\bneq \label{fieldcommrel}
\begin{split}
&\Phi_{R,r}(f)\widehat{\Phi}^*_{R,r}(g) - e^{i\mu}\widehat{\Phi}^*_{R,r}(g)\Phi_{R,r}(f) = \\
e^{i\mu}\int &d\theta g^+(\theta)f^-(\theta) T_{r,R}(\theta)^2 - e^{-2i\rho}\int d\theta f^+(\theta) g^-(\theta) T_{\bar{R},\bar{r}}(\theta)^2.
\end{split}
\eneq
For spacelike separated $f$ and $g$ with $supp(g)$ to the left of $supp(f)$ we want the right-hand side of equation \eqref{fieldcommrel} to vanish; thus we need to set $\rho = -\frac{\mu}{2}$ such that $e^{i\mu} = e^{-2i\rho}$. Then we can use the arguments in \cite{Lechner03} to show that
\bneq
(\Phi_{R,r}(f)\widehat{\Phi}^*_{R,r}(g) - e^{i\mu}\widehat{\Phi}^*_{R,r}(g)\Phi_{R,r}(f))\Psi_n^m
\eneq
vanishes for all $\Psi_n^m \equiv \Psi_n\otimes\Psi_m$ in the domain of definition of $\Phi$ if the deformation functions $R$ and $r$ are analytic in the strip $S(0,\pi):=\{z\in\mathds{C}: 0<Im(z)<\pi\}$, bounded and continuous on its closure and satisfy the ``crossing relations''
\bneq
R(\theta+i\pi) = \overline{r(\theta)} \ , \hspace{1em} \ r(\theta+i\pi) = \overline{R(\theta)}.
\eneq
These conditions then allow us to shift the integration in \eqref{fieldcommrel} from $\theta$ to $\theta+i\pi$, because of the known analyticity properties of $f^\pm$ and $g^\pm$. For a more detailed argument see \cite{Lechner03} or the proof of Proposition 3.1. \\
Therefore we have shown that if there is a wedge $W$ such that $supp(f)\subset W$ and $supp(g)\subset W'$ the fields satisfy
\bneq
\Phi_{R,r}(f)\widehat{\Phi}^*_{R,r}(g) = e^{i\mu}\widehat{\Phi}^*_{R,r}(g)\Phi_{R,r}(f).
\eneq \vspace{1em}

To summarize our construction let us compare it with the neutral case studied in \cite{Lechner12} by writing down the input we need in both cases to define a deformation. \\
\begin{itemize}
\item In \cite{Lechner12} a deformation was defined on the neutral bosonic Fock space by choosing a function $\mathcal{R}: S(0,\pi) \rightarrow \mathds{C}$, which is analytic in $S(0,\pi)$, bounded on $\overline{S(0,\pi)}$ and satisfies
\bneq \label{condR}
\mathcal{R}(-x) = \overline{\mathcal{R}(x)} = \mathcal{R}(x)^{-1}\ , \hspace{1em} \forall x\in\mathds{R}
\eneq
and the crossing relation
\bneq \label{crossing}
\mathcal{R}(i\pi-x) = \mathcal{R}(x)\ , \hspace{1em} \forall x\in\mathds{R}.
\eneq
The most general class of such functions has been calculated in \cite{Lechner06} and it turns out that they are of the form
\bneq
\mathcal{R}(\theta) = \pm e^{ia\sinh{\theta}}\prod_k \frac{\sinh{\beta_k}-\sinh{\theta}}{\sinh{\beta_k}+\sinh{\theta}}
\eneq
with some parameters $a,\{\beta_k\}$ satisfying certain additional conditions.
\item In the charged case at hand we now have two functions $R$ and $r$, analytic in $S(0,\pi)$ and bounded on its closure, which have to satisfy relations similar to \eqref{condR} and \eqref{crossing}, namely
\bneq
\begin{split}
e^{-i\mu}R(-x) &= \overline{R(x)} = R(x)^{-1}\ , \hspace{1em} \forall x\in\mathds{R} \\
e^{i\mu}r(-x) &= \overline{r(x)} = r(x)^{-1}\ , \hspace{1em} \forall x\in\mathds{R}
\end{split}
\eneq 
and the crossing relation in this case turns out to be
\bneq \label{doublecrossing}
R(i\pi-x) = r(x)\ , \hspace{1em} \forall x\in\mathds{R}.
\eneq
Now of course one can always separate the phase factor $e^{\pm i\mu}$ from the deformation functions by defining $R =: e^{i\frac{\mu}{2}} R^+$ and $r =: e^{-i\frac{\mu}{2}} R^-$. The functions $R^\pm$ then satisfy the usual relations \eqref{condR} without the phase factors present. But the aforementioned conditions can be simplified further by noting that because of the analyticity of $R^+$ and $R^-$ the crossing relation \eqref{doublecrossing} can be used to \emph{define} the function $R^-$ in terms of $R^+$ by setting $R^-(\theta) := R^+(i\pi-\theta), \forall \theta\in S(0,\pi)$. \\
Therefore we are left with choosing a parameter $\mu\in\mathds{R}$ and a single deformation function $R^+: S(0,\pi)\rightarrow\mathds{C}$, satisfying
\bneq
\begin{split}
&R^+(-x) = \overline{R^+(x)} = R^+(x)^{-1}\ , \hspace{1em} \forall x\in\mathds{R} \\
&R^+(i\pi-x) = \overline{R^+(i\pi+x)} = R^+(i\pi+x)^{-1}\ , \hspace{1em} \forall x\in\mathds{R}
\end{split}
\eneq
but \emph{not} the crossing relation \eqref{crossing}! So there is no condition relating the values of $R^+$ on the upper boundary of $S(0,\pi)$ to those on the real boundary.
\end{itemize}
\vspace{1em}
It is now an interesting question if the class of admissible deformation functions is in the charged case actually larger than in the neutral case. The answer to this question is yes. To show that there really are functions $f: \overline{S(0,\pi)}\rightarrow \mathds{C}$ satisfying all the above requirements for our $R$-functions but \emph{not} the crossing symmetry $f(i\pi-x) = f(x), \forall x\in\mathds{R}$, consider the functions \footnote{I want to thank Christian Köhler and Martin Könenberg, for pointing out this kind of functions to me.}
\bneq
f(z) := i \frac{e^{z}\alpha-i\bar{\alpha}}{e^{z}\bar{\alpha}+i\alpha}  , \hspace{1em} \alpha = 1-i w \ , \ \ w\in\mathds{R}, |w|<1.
\eneq
They are clearly analytic in $S(0,\pi)$ and a short calculation shows that for real $x$ they also satisfy $f(-x) = \overline{f(x)} = f(x)^{-1}$ and $f(i\pi-x) = \overline{f(i\pi+x)} = f(i\pi+x)^{-1}$. But the condition $f(i\pi-x) = f(x)$ is  not satisfied (for $w=0$ at least $f(i\pi-x) = -f(x)$ still holds)! \\
So we could chose e.g. an arbitrary deformation function $\mathcal{R}$ from the neutral case and define
\bneq \label{generaldeformation}
R^+(z) := f(z) \mathcal{R}(z) \ , \hspace{1em} R^-(z) = f(i\pi-z) \mathcal{R}(z).
\eneq
We believe that this is essentially the most general possibility of choosing the function $R^+$ but we could not yet prove this statement. \\ \vspace{1em}

In the following we will see how these deformation functions for the charged fields are related to the two-particle scattering matrix of our model. Because we are now only interested in the momentum dependence of the S-matrix we set $\mu=0$ which means that the deformed fields commute if the test functions have the right support properties. \\
Because of the simple structure of our deformed fields and because everything is on-shell in our setting the outgoing scattering states are of the form $a^*_{R,r}(f^+)a^*_{\bar{R},\bar{r}}(g^+)\Omega$ if the Fourier transforms of $f$ and $g$ have compact support and $\operatorname{supp}(g)$ is to the left of $\operatorname{supp}(f)$. (For the exact definition of the outgoing/incoming scattering states see section \ref{sec:scattering}) So because of $T_r(\theta)\Omega = \Omega$ the S-matrix for particle-particle ($S_{pp}$) and antiparticle-antiparticle ($S_{aa}$) scattering is formally just $R^2$. But in the charged case at hand we also have to consider states of the form $a^*_{R,r}(f^+)b^*_{\bar{R},\bar{r}}(g^+)\Omega$ which basically look like
\beq
(a^*_{R,r}(f^+)b^*_{\bar{R},\bar{r}}(g^+)\Omega)(\theta_1,\theta_2) \sim r(\theta_1-\theta_2)f^+(\theta_1)g^+(\theta_2)+r(\theta_2-\theta_1)f^+(\theta_2)g^+(\theta_1).
\eeq
Therefore the S-matrix for particle-antiparticle scattering ($S_{pa}$) turns out to be $r^2$. We can summarize this by formally writing
\bneq
\begin{split}
S_{pp} \sim S_{aa} \sim R^2, \\
S_{pa} \sim S_{ap} \sim r^2.
\end{split}
\eneq
Recalling the relation $r(\theta) = \overline{R(\theta+i\pi)}$ we see that $R^2$ evaluated at the lower boundary of the strip $S(0,\pi)$ determines the scattering between the same kind of particles while $R^2$ at the upper boundary determines particle-antiparticle scattering.
\\ \vspace{1em}

In the next step we want to analyze the dependence on the additional parameter $\mu$. Setting $\mu= 2\pi \lambda$, $\lambda \in\mathds{R}$ the deformed fields satisfy Anyonic commutation relations in the usual form, i.e.
\bneq
\Phi \widehat{\Phi} = e^{-2i\pi\lambda} \widehat{\Phi} \Phi \ , \hspace{1em} \ \Phi \widehat{\Phi}^* = e^{2i\pi\lambda} \widehat{\Phi}^* \Phi.
\eneq
In this case we can simplify the above deformation and rewrite it using the charge operator $Q$. For this purpose we take for simplicity a standard deformation function $\mathcal{R}$ satisfying $\mathcal{R}(-\theta) = \overline{\mathcal{R}(\theta)}$ and define
\bneq \label{simplefunctions}
R(\theta) = e^{i\pi\lambda}\mathcal{R}(\theta)\ \hspace{1em} \ r(\theta) = e^{-i\pi\lambda}\mathcal{R}(\theta). 
\eneq
The corresponding multiplication operator $T_{R,r}(\theta) = e^{-i\pi\lambda/2} T_R(\theta)\otimes T_r(\theta)$ can then be written as 
\bneq
T_{R,r}(\theta) = (T_{\mathcal{R}}(\theta)\otimes T_{\mathcal{R}}(\theta) )e^{i\pi\lambda(Q-1/2)}.
\eneq
Note that by explicitly using the charge operator $Q$ in the deformation we see that we are effectively using a different deformation function on every charge sector, i.e. $T_{R,r}$ depends on the charge of the vector we are applying it to (albeit in a rather trivial manner in the above example). \\
We can now again define the deformed fields according to
\bneq \label{fielddef}
\begin{split}
\Phi_{\mathcal{R},\lambda}(f) &:= a^*_{\mathcal{R},\lambda}(f^+) + b_{\mathcal{R},\lambda}(\bar{f}^+) = \Phi_{\mathcal{R}}(f)e^{-i\pi\lambda(Q+1/2)}, \\
\Phi^*_{\mathcal{R},\lambda}(f) &= b^*_{\mathcal{R},\lambda}(f^+) + a_{\mathcal{R},\lambda}(\bar{f}^+) = \Phi^*_{\mathcal{R}}(f)e^{i\pi\lambda(Q-1/2)},
\end{split}
\eneq
and calculate the fields for the opposite wedge,
\beq
J \Phi_{\mathcal{R},\lambda}(\alpha_j f) J = \Phi_{\overline{\mathcal{R}},-\lambda}(f).
\eeq
Using the commutation relations 
\beq
Q \Phi_{\mathcal{R},\lambda}(f) = \Phi_{\mathcal{R},\lambda}(f) (Q+1) \ , \hspace{1em} Q \Phi^*_{\mathcal{R},\lambda}(f) = \Phi^*_{\mathcal{R},\lambda}(f)(Q-1)
\eeq
between the fields and the charge operator it is now easy to check that the fields satisfy
\bneq \label{anyoniccommrel}
\begin{split}
\Phi_{\mathcal{R},\lambda}(f) \Phi^*_{\overline{\mathcal{R}},-\lambda}(g) &= e^{2\pi i \lambda}\, \Phi^*_{\overline{\mathcal{R}},-\lambda}(g)\Phi_{\mathcal{R},\lambda}(f) , \\
\Phi_{\mathcal{R},\lambda}(f) \Phi_{\overline{\mathcal{R}},-\lambda}(g) &= e^{-2\pi i \lambda}\, \Phi_{\overline{\mathcal{R}},-\lambda}(g)\Phi_{\mathcal{R},\lambda}(f),
\end{split}
\eneq
if $f$ and $g$ have again the right support properties. \\
Thus we have seen that if we apply the deformations of \cite{Lechner12} to a charged scalar field, we can change the deformations using the charge operator to obtain wedge-localized anyonic one-particle generators. \\
From the definition \eqref{fielddef} we also see that the one-particle states the fields create are changed by a constant factor $e^{- i\pi \lambda/2}$ and also the S-matrix elements get multiplied with such exponential factors. Furthermore, the scattering states are no longer symmetric under permutations but inherit the braided symmetry from the fields which create them from the vacuum.  \\ \vspace{1em}

However, one could get the same result by choosing a deformation with $\mu=0$ and instead take a representation $J_{\lambda}$ of the reflections with $\beta = \pi\lambda q$, i.e.
\bneq
J_{\lambda} := e^{i\pi\lambda Q^2} J = \bigoplus_{q=-\infty}^\infty e^{i\pi\lambda q^2} J P_q,
\eneq
where $J$ is just the representation used before, acting as complex conjugation, and $P_q$ is the projection onto the charge $q$ Hilbert space. \\
Now take any wedge-localized charged bosonic field $\phi$ and define the field for the opposite wedge as 
\bneq
\hat{\phi} := J_\lambda \phi J_\lambda.
\eneq
A straightforward calculation then shows that these fields indeed satisfy commutation relations of the form \eqref{anyoniccommrel}. This construction is possible because we can choose a different representation of the Lorentz group for every charge $q$ due to the charge structure of the Hilbert space. In this way it is possible to arbitrarily choose the commutation relations of the wedge-local fields, which is in accordance with the well-known fact that the statistics of a quantum field is not an intrinsic concept in $d=1+1$ dimensions (cf. ``Bosonization'', ``Fermionization''). \\
Hence we will now proceed to a more interesting construction, namely wedge-local fields with anyonic statistics in $d=2+1$. \\

\section{Deformations on Three-Dimensional Minkowski Space}
In the next step we will try to find a wedge-local deformation leading to braided commutation relations in $d=2+1$. The reason why this is considerably more complicated is the presence of the rotations in the Lorentz group in three dimensions and therefore there are not only left- and right wedges but a continuous family of possible directions of wedges. Moreover, it is known that for the fields to have definite commutation relations they have to carry further information in addition to the localization region. Therefore, we consider localization not only in wedges but in so-called paths of wedges, containing as an additional information a kind of winding number (for the definition of ``paths of wedges'' used in this work see e.g. \cite{Mund03}). Such paths of wedges are denoted by $\tW = (W,\tilde{e})$ and are defined by a wedge $W$ and a homotopy class $\tilde{e}$ of paths in the manifold of space-like directions $H^3$, starting at a reference direction $e_0$ and ending at a point inside the wedge. To be concrete we chose $e_0$ as $e_0 = (0,0,-1)$. \\
As in the two-dimensional case we will work on the charged Fock-space, where we now have $\mH_1^{\pm} = L^2(\mathds{R}^3,d\mu)$ with the measure on the mass shell $d\mu(p) = d^3p \delta(p^2-m^2)\theta(p_0)$ and we use the shorthand notation $\Psi_n^m(\up) := \Psi_n^m(p_1,...,p_{n+m})$ with $\up = (p_1,...,p_{n+m})$. \\ 

\subsection{Representation of the Covering Group of $\mathcal{L}_+^\uparrow$}
Because of the spin-statistics theorem in 2+1 dimensions \cite{Mund09} we want $2\pi$-rotations to act non-trivially, i.e. $U(2\pi) \neq \pm 1$. We therefore consider representations of the universal covering of the Poincaré group $\tilde{\mathcal{P}}_+^\uparrow$, which is the semi-direct product of the translations with the universal covering of the Lorentz group $\tilde{\mathcal{L}}_+^\uparrow$. In three dimensions this group can be identified with the set 
\bneq
\{ (\gamma,\omega) | \gamma \in \mathds{C}, |\gamma|<1, \omega \in \mathds{R}\},
\eneq
with corresponding group multiplication $(\gamma_1,\omega_1)(\gamma_2,\omega_2) = (\gamma_3,\omega_3)$, which is given by \cite[p.594]{Bargmann48}
\beq
\begin{split}
\gamma_3 &= (\gamma_2 + \gamma_1 e^{-i\omega_2})(1+\gamma_1\overline{\gamma_2} e^{-i\omega_2})^{-1} ,\\
\omega_3 &= \omega_1+\omega_2 - i\log\left[(1+\gamma_1\overline{\gamma_2}e^{-i\omega_2})(c.c.)^{-1}\right].
\end{split}
\eeq
Identifying elements of this group with homotopy classes $\tL$ of paths $t\mapsto \Lambda(t) \in \mathcal{L}_+^\uparrow$ starting at the unit element and ending at $\Lambda \in \mathcal{L}_+^\uparrow$ we can define the action of $\tilde{\mathcal{L}}_+^\uparrow$ on paths of wedges $\tW = (W,\tilde{e})$ according to $\tL\cdot \tW = (\Lambda W, \tL\cdot\tilde{e})$.
\\
The translations act on the Fock-space in the usual way as $(U(a)\Psi)(p) = e^{iap}\Psi(p)$ and they obviously do not change winding numbers of localization regions. Therefore we will be interested in the representation of the covering of the Lorentz group most of the time. Because of the charge structure of our Fock-space, $\mH = \oplus \mH_q$, we consider a representation of $\tilde{\mathcal{P}}_+^\uparrow$ of the form $U = \oplus U_q$. The observables of a theory should commute with rotations around $2\pi$ and their restriction onto a subspace $\mH_q$ with fixed charge should be irreducible. Therefore the $2\pi$-rotations act as a multiple of the identity on vectors with fixed charge, i.e. 
\bneq
U_q(2\pi) = e^{2\pi i \mathcal{S}_q}\cdot \mathds{1},
\eneq
where $\mathcal{S}_q$ is called the spin of the sector with charge $q$ (determined only modulo 1). From the general theory of superselection sectors one knows that the spin is the same for a sector and its conjugate sector (see \cite{FroeMar, Mund09}), therefore we must have $\mathcal{S}_q = \mathcal{S}_{-q}$. When restricted to the 1-(anti-) particle Hilbert space $\mH_1^{\pm}$ we also want our representation of the Poincaré group to be one of the well-known irreducible unitary representations for a spin $\sigma$, defined according to,
\bneq
U^{(\sigma)}(a,\tL)\varphi)(p) := e^{i a\cdot p}\, e^{i\sigma\Omega(\tL,p)}\varphi(\Lambda^{-1}p).
\eneq
The factor $\Omega(\tL,p)$ is the Wigner-rotation, which can be expressed for $\tL = (\gamma,\omega)$ according to (see e.g. \cite[Appendix B]{Mund03})
\beq
\Omega(\tL,p) = \omega - i\log\left[(1-\gamma(p)\bar{\gamma}e^{-i\omega})(c.c)^{-1}\right] - i\log\left[\left(1+\frac{\gamma-\gamma(p)e^{-i\omega}} {1-\gamma(p)\bar{\gamma}e^{-i\omega}}\bar{\gamma}(\Lambda^{-1}p)\right). (c.c)^{-1}\right].
\eeq 
It satisfies the cocycle relation
\bneq \label{cocyclerel}
\Omega(\tL\tL ',p) = \Omega(\tL,p) + \Omega(\tL ',\Lambda^{-1}p)
\eneq
and for pure rotations $\tilde{r}$ it simplifies to $\Omega(\tilde{r}(\omega),p) = \omega$. \\
Motivated by the above considerations we define the full representation $U$ on $\mH$ according to
\bneq \label{representation}
(U(a,\tL)\Psi)_n^m(\up) := e^{i a\cdot \up}\ e^{i s_q \Omega_n^m(\tL,\up)}\Psi_n^m(\Lambda^{-1}\up),
\eneq
where we have introduced the notation
\bneq
\Omega_n^m(\tL,\up) := \sum_{i=1}^n\Omega(\tL,p_i)-\sum_{j=n+1}^{n+m}\Omega(\tL,p_j).
\eneq
From this we see that a $2\pi$-rotation acts on vectors of charge $q$ according to
\beq
U_q(2\pi) = e^{2\pi i q s_q}\cdot\mathds{1} \overset{!}{=} e^{2\pi i \mathcal{S}_q}\cdot \mathds{1},
\eeq
leading to $\mathcal{S}_q = q s_q$ which implies that $s_{-q} = -s_q$. A simple choice for $s_q$ is $s_q = \lambda q$ with an unspecified parameter $\lambda\in\mathds{R}$ and we will see that this choice leads to anyonic statistics for our deformed fields. \\
Restricting this representation to $\mH_1$ leads to
\beq
\begin{split}
(U(a,\tL)\Psi)_1^0(p) &= e^{i a\cdot p}\, e^{is_1\Omega(\tL,p)}\Psi_1^0(\Lambda^{-1}p)= e^{i a\cdot p}\, e^{i\lambda\Omega(\tL,p)}\Psi_1^0(\Lambda^{-1}p), \\
(U(a,\tL)\Psi)_0^1(p) &= e^{i a\cdot p}\, e^{-is_{-1}\Omega(\tL,p)}\Psi_0^1(\Lambda^{-1}p) = e^{i a\cdot p}\, e^{i\lambda\Omega(\tL,p)}\Psi_0^1(\Lambda^{-1}p),
\end{split}
\eeq
so we see that on $\mH_1$ the representation really reduces to
\beq
U(a,\tL)\mid_{\mH_1} = U^{(\lambda)}(a,\tL)\oplus U^{(\lambda)}(a,\tL).
\eeq

\textbf{Extension to $\tilde{\mathcal{P}}_+$:} \\
Given the proper orthochronous Poincaré group $\mathcal{P}_+^\uparrow$ one can obtain the proper Poincaré group $\mathcal{P}_+$ by adjoining the reflection $j$ at the $x_2$-axis, which satisfies the relations
\bneq
j^2 = 1 \ , \hspace{1em} j\Lambda_1(t)j = \Lambda_1(t) \ , \hspace{1em} jr(\omega)j = r(-\omega) \ , \hspace{1em} j (a,1) j = (j\cdot a,1),
\eneq 
(which then imply $j\Lambda_2(t)j = \Lambda_2(-t)$), where $\Lambda_1, \Lambda_2$ are boosts in the direction of the $x_1, x_2$ axis respectively. This yields a disconnected group and its universal covering $\tilde{\mathcal{P}}_+$ can be defined by adjoining an element $\tilde{j}$ to $\tilde{\mathcal{P}}_+^\uparrow$ with the relations
\bneq
\tilde{j}^2 = 1 \ , \hspace{1em} \tilde{j} (a,(\gamma,\omega)) \tilde{j} = (j\cdot a,(\bar{\gamma},-\omega)).
\eneq
Defining by $\tilde{j}\cdot\tilde{e}$ the equivalence class w.r.t. $j\cdot W$ of the path $t \mapsto j\cdot e(t)$ the element $\tilde{j}$ acts on $\tW$ according to
\bneq \label{jaction}
\tilde{j}\cdot\tW := (j\cdot W, \tilde{j}\cdot\tilde{e}),
\eneq
where $\tilde{j}\cdot\tilde{e}$ is still a path starting at the reference direction $e_0$ because we have chosen $e_0$ invariant under $j$.
\vspace{1em}

\subsection{The Deformed Model}
On the Hilbert space $\mH$ we again define the free charged scalar field $\Phi(f) = a^*(f^+) + b(\bar{f}^+)$ with $a$ and $b$ defined as in the previous chapter and $f^\pm(p) = \int dx f(\pm x) e^{ixp}$. This field is local and covariant w.r.t. the representation $U$ only for $\lambda = 0$. Now we want to deform this field using multiplicative deformations such that the deformed field $\Phi_{\tW}(f)$ is covariant w.r.t. $U$ for an arbitrary $\lambda\in\mathds{R}$ and localized in $supp f + \tW$, where for simplicity we only consider wedges having the origin contained in their edge because of translation covariance. Taking a $\tW' = (W',\tilde{e})$ with an arbitrary path $\tilde{e}$ ending in $W'$ and two test-functions $f,g$ such that $supp(f)+W$ is causally separated from $supp(g)+W'$ we want the fields to satisfy
\bneq
\Phi_{\tW}(f) \Phi^\sharp_{\tW'}(g) = e^{\mp 2\pi i \lambda k(\tW,\tW')} \Phi^\sharp_{\tW'}(g) \Phi_{\tW}(f).
\eneq
The statistics factor only depends on the two wedges $\tW, \tW'$ and we will see that it is related to the relative winding number $N(\tW,\tW ')$ defined below according to $k(\tW,\tW') = -2 N(\tW,\tW') - 1$. \\
General multiplication operators on $\mathcal{H}$ are of the form
\bneq
(T_{\tW}(p)\Psi)_n^m(\up) := A_{\tW}^{n,m}(p;\up) \Psi_n^m(\up),
\eneq
and they are used to deform the creation and annihilation operators,
\bneq
\begin{split}
a_{\tW}(p) &:= T_{\tW}(p) a(p) \ , \hspace{6.4em} a^*_{\tW}(p) := a_{\tW}(p)^* = a^*(p)T_{\tW}(p)^*, \\
b_{\tW}(p) &:= C a_{\tW}(p) C \equiv T^c_{\tW}(p) b(p) \ , \hspace{1em} b^*_{\tW}(p) := b_{\tW}(p)^*.
\end{split}
\eneq
The charge conjugated multiplication operator acts according to
\bneq
(T^c_{\tW}(p)\Psi)_n^m(\up) = A_{\tW}^{m,n}(p;\up^c) \Psi_n^m(\up),
\eneq
with $\up^c = (p_{n+1},...,p_{n+m},p_1,...,p_n)$. These definitions mean that for a one-particle vector $\varphi$ the annihilation operator e.g. acts as
\beq
(a_{\tW}(\varphi)\Psi)_n^m(\up) = \sqrt{n+1}\int d\mu(p)\overline{\varphi(p)} A_{\tW}^{n,m}(p;\up)\Psi_{n+1}^m(p,\up).
\eeq
The deformed field is then defined as $\Phi_{\tW}(f) = a^*_{\tW}(f^+)+ b_{\tW}(\bar{f}^+)$ and we want it to be covariant under $\tilde{\mathcal{P}}_+^\uparrow$, i.e.
\bneq
U(a,\tL) \Phi_{\tW}(f) U(a,\tL)^{-1} = \Phi_{\tL\tW}(\alpha_{(a,\Lambda)}(f)) \ , \hspace{2em} (\alpha_{(a,\Lambda)}(f))(x) = f(\Lambda^{-1} (x-a)).
\eneq
Therefore also the deformed annihilation operators have to satisfy $U(a,\tL) a_{\tW}(\varphi) U(a,\tL)^{-1} = a_{\tL\tW}(\alpha_{(a,\Lambda)}(\varphi))$ which leads to the following relation for the deformation functions,
\bneq
e^{-i\lambda \Omega_n^m(\tL,\up)}e^{-i\lambda(q+1)\Omega(\tL,p) }A_{\tW}^{n,m}(\Lambda^{-1}p;\Lambda^{-1}\up) = A_{\tL\tW}^{n,m}(p;\up).
\eneq
The product structure of this covariance condition and the known deformation functions in \cite{Lechner12} now motivate the following ansatz,
\bneq
A_{\tW}^{n,m}(p;\up) := u_{\tW}^\lambda(p)^{q+1}\prod_{i=1}^{n}u_{\tW}^\lambda(p_i)R(Qp\cdot p_i)\prod_{j=n+1}^{n+m}\overline{u_{\tW}^\lambda(p_j)} R(Qp\cdot p_j) ,
\eneq
where the functions $u$ and $R$ are defined in the following way. 
\begin{itemize}
\item $R$ is a ``standard'' deformation function in the sense of \cite{Lechner12} as already described in the two-dimensional case, i.e. it satisfies $R(-a) = \overline{R(a)} = R(a)^{-1}$ and it has an analytic continuation into the upper half plane, continuous on its closure, to guarantee the right commutation relations of the deformed fields for space-like separation. The $Q \equiv Q(W)$ in the argument of $R$ is a $W$ dependent ``deformation matrix'', which is \emph{anti-symmetric} w.r.t. the Lorentz inner product and its definition can be found e.g. in \cite{GrossLech07, GrossLech}. Denoting by $L_W$ a Lorentz transformation connecting $W$ and $W_0$, i.e. $W = L_W W_0$ \footnote{Such a $L_W$ always exists because the family of wedges has been defined as the orbit of $W_0$ under $\mathcal{L}_+^\uparrow$.}, the matrix $Q(W)$ is defined according to $Q(W) = L_W Q_0 L_W^{-1}$ where 
\bneq
Q_0 = \kappa \begin{pmatrix} 0 & 1 & 0 \\ 1 & 0 & 0 \\ 0 & 0 & 0 \end{pmatrix} \ , \hspace{1em} \kappa > 0
\eneq
is the matrix belonging to the standard wedge $W_0$. This dependence on $W$ of the matrix $Q$ ensures the right covariance properties of the deformation functions and it follows in particular that $Q(W') = -Q(W)$.

\item The functions $u_{\tW}^\lambda$ are intertwiners for the representation $U^{(\lambda)}$ of $\mathcal{L}_+^\uparrow$ which have to satisfy the relation
\bneq \label{intrel}
e^{-i\lambda\Omega(\tL,p)}u_{\tW}^\lambda(\Lambda^{-1}p) = u_{\tL\tW}^\lambda(p).
\eneq
To find such functions $u$ satisfying this condition we first define the path for the standard wedge $\tW_0 = (W_0,\tilde{e}_0)$ with $W_0 = \{x\in\mathds{R}^3\mid x_1>|x_0|\}$ and $\tilde{e}_0$ is a path starting at $e_0$ and staying inside $W_0$.  We then consider equation \eqref{intrel} for $\tW_0$ and $\tL = \tilde{L}_{\tW}$, where $\tilde{L}_{\tW}$ is a Lorentz transformation connecting $\tW$ and $\tW_0$, i.e. $\tilde{L}_{\tW}\tW_0 = \tW$. We then get
\bneq
u_ {\tW}^\lambda(p) = e^{-i\lambda\Omega(\tilde{L}_{\tW},p)} u_{\tW_0}^\lambda(\tilde{L}_{\tW}^{-1}p),
\eneq
which shows that the intertwiner for $\tW$ is determined by the intertwiner $u_0^\lambda := u_{\tW_0}^\lambda$ for the standard wedge $\tW_0$. But this construction is not unique, because for every $\tilde{L}_{\tW}$ also $\tilde{L}_{\tW}\Lambda_1(t)$ is a $\mathcal{L}_+^\uparrow$ transformation mapping $\tW_0$ to $\tW$, because the boosts $\Lambda_1$ in 1-direction leave the standard wedge $\tW_0$ invariant. To restore uniqueness the functions $u_0^\lambda$ need to satisfy the consistency condition
\bneq
u_0^\lambda(p) = e^{-i\lambda\Omega(\Lambda_1(t),p)} u_0^\lambda(\Lambda_1(-t)p).
\eneq
To construct such functions we observe that, according to \cite[Lemma C.1]{Mund03}, the Wigner rotation factor $e^{-i\lambda\Omega(\Lambda_1(t),p)}$ can be written as
\bneq
e^{-i\lambda\Omega(\Lambda_1(t),p)} = v(p)^\lambda v(\Lambda(-t)p)^{-\lambda} \ , \hspace{2em} v(p) := \frac{p_0+m-p_1+ip_2}{p_0+m-p_1-ip_2}.
\eneq
This leads to the solution
\bneq
u_0^\lambda(p) := f(p_2)^\lambda v(p)^\lambda
\eneq
for $u_0^\lambda$ where $f$ is a yet undefined function of $p_2$ and thus invariant under $x_1$-boosts.\footnote{Note that every function on the mass shell which is invariant under boosts in the 1-direction is a function of $p_2$.} One can now easily check that our intertwiner function $u_{\tW}^\lambda(p) := e^{-i\lambda\Omega(\tilde{L}_{\tW},p)} u_0^\lambda(L_{\tW}^{-1}p)$ satisfies relation \eqref{intrel}. \\
Note that (apart from $R$ which determines the momentum dependence of the S-matrix) the deformation function $A_{\tW}^{n,m}$ is fixed by covariance up to the function $f(p_2)$ which has to be chosen in such a way that the creation and annihilation operators satisfy the right commutation relations.
\end{itemize}
\vspace{1em}
To calculate these relations we first need commutation relations between undeformed creation{\slash}annihilation operators and the deformation operators $T_{\tW}$. Acting on the charge $q$ Hilbert space we get e.g.
\bneq
\begin{split}
\left.T_{\tW}(p) a(p') \right|_q &= \frac{u_{\tW}^\lambda(p)^{q-1}}{u_{\tW}^\lambda(p)^{q}} (u_{\tW}^\lambda(p') R(Qp\cdot p'))^{-1} \left.a(p') T_{\tW}(p)\right|_q \\
&= u_{\tW}^\lambda(p)^{-1} u_{\tW}^\lambda(p')^{-1} R(Qp\cdot p')^{-1} a(p') T_{\tW}(p) =: B_{\tW}(p,p') a(p') T_{\tW}(p), \\
T_{\tW}(p) a^*(p') &= B_{\tW}(p,p')^{-1} a^*(p') T_{\tW}(p).
\end{split}
\eneq
In the same way we can calculate
\bneq
T_{\tW}(p) b(p') = u_{\tW}^\lambda(p) \overline{u_{\tW}^\lambda(p')}^{-1} R(Qp\cdot p')^{-1} b(p') T_{\tW}(p) =: C_{\tW}(p,p') b(p') T_{\tW}(p).
\eneq
All other relations now follow by charge conjugation and taking adjoints. Using these commutation relations we now obtain
\bneq
\begin{split}
a_{\tW}(p) a_{\tW '}(p') &= T_{\tW}(p) a(p) T_{\tW '}(p') a(p')\\
&= B_{\tW '}(p',p)^{-1} B_{\tW}(p,p') a_{\tW '}(p') a_{\tW}(p).
\end{split}
\eneq
Inserting the definition of $B$ we get
\bneq
B_{\tW '}(p',p)^{-1} B_{\tW}(p,p') = \frac{u_{\tW '}^\lambda(p') u_{\tW '}^\lambda(p) R(Qp'\cdot p)}{u_{\tW}^\lambda(p) u_{\tW}(p') R(Qp\cdot p')} =  \frac{u_{\tW '}^\lambda(p)}{u_{\tW}^\lambda(p)} \frac{u_{\tW '}^\lambda(p')}{u_{\tW}^\lambda(p')},
\eneq
where the last equation follows by using the relation $R(-a) = R(a)^{-1}$. From the definition of the $u$'s we see that
\bneq \label{uoveru}
\frac{u_{\tW '}^\lambda(p)}{u_{\tW}^\lambda(p)} = e^{-i\lambda (\Omega(\tilde{L}_{\tW '},p) - \Omega(\tilde{L}_{\tW},p))} u_0^\lambda(L_{\tW '}^{-1} p) u_0^\lambda(L_{\tW}^{-1}p)^{-1},
\eneq
and we are going to show that this term equals $e^{-i\pi\lambda k(\tW,\tW')}$ with a $k\in\mathds{Z}$, depending only on the wedges $\tW, \tW'$ and not on the momenta $p, p'$. \\
To calculate this expression we need to look more closely at the relation between $\tilde{L}_{\tW}$ and $\tilde{L}_{\tW '}$. First note that if $W = L_W W_0$ then $W' = L_W W_0' = L_W r(\pi) W_0$. Lifting this to paths of wedges we see that for every $\tW'$ there is an odd number $k\in 2\mathds{Z}+1$ such that $\tW' = \tilde{L}_{\tW} \tilde{r}(k\pi) \tW_0$.
This shows that every $\tilde{L}_{\tW '}$ is of the form $\tilde{L}_{\tW '} = \tilde{L}_{\tW} \tilde{r}(k\pi) \Lambda_1(t)$, where $\Lambda_1$ is a boost in $x_1$ direction. Using that the intertwiners do not depend on the choice of this $\Lambda_1(t)$ we obtain
\beq
\frac{u_{\tW '}^\lambda(p)}{u_{\tW}^\lambda(p)} = e^{-i\lambda (\Omega(\tilde{L}_{\tW}\tilde{r}(k\pi),p) - \Omega(\tilde{L}_{\tW},p))} u_0^\lambda(r(-k\pi)L_{\tW}^{-1} p) u_0^\lambda(L_{\tW}^{-1}p)^{-1}.
\eeq
To calculate the exponential factor we need the cocycle relation \eqref{cocyclerel} of the Wigner rotation factors
\beq
\Omega(\tL \tL ',p) = \Omega(\tL,p) + \Omega(\tL ',\Lambda^{-1}p).
\eeq
This leads to
\bneq
\begin{split}
\Omega(\tilde{L}_{\tW}\tilde{r}(k\pi),p) - \Omega(\tilde{L}_{\tW},p) &= \Omega(\tilde{L}_{\tW},p) + \Omega(\tilde{r}(k\pi),L_{\tW}^{-1}p)-\Omega(\tilde{L}_{\tW},p) \\
&= \Omega(\tilde{r}(k\pi),L_{\tW}^{-1}p) = k\pi,
\end{split}
\eneq
where $k \equiv k(\tW,\tW ')$ obviously only depends on the winding number of $\tW '$ w.r.t $\tW$. Such a winding number, $N(\tilde{C}_1,\tilde{C}_2)$, can be defined for general causally separated spacelike cones\footnote{A spacelike cone $C = a + \cup_{\nu\geq 0}\nu \mathcal{O}$ is defined through its apex $a$ and a double cone $\mathcal{O}$, which is spacelike separated from $a$.} $\tilde{C}_1$ and $\tilde{C}_2$ (see e.g. \cite{BrosMund}) in the following way. Let $d\theta$ be the angle one-form in some fixed Lorentz frame, and for a path $\tilde{C} = (C,\tilde{e})$ let $\theta(\tilde{C})$ be the set of corresponding ``accumulated angles'', namely the interval
\bneq
\theta(\tilde{C}) := \left\{\int_e d\theta : e \in \tilde{e}\right\}.
\eneq
Now given two paths $\tilde{C}_1, \tilde{C}_2$ with $C_1$ causally separated from $C_2$, one can define the relative winding number $N(\tilde{C}_1,\tilde{C}_2)$ of $\tilde{C}_2$ w.r.t. $\tilde{C}_1$ to be the unique integer $n$
such that
\bneq
\theta(\tilde{C}_2) + 2\pi n < \theta(\tilde{C}_1) < \theta(\tilde{C}_2) + 2\pi(n+1).
\eneq
Considering two wedges $\tW$ and $\tW '$  one can proof the following relation between $k(\tW,\tW ')$ and $N(\tW, \tW ')$.
\begin{lemma}
Let $k(\tW,\tW ')$ and $N(\tW,\tW ')$ be defined as before. Then the relation
\bneq
-k(\tW,\tW ') = 2 N(\tW,\tW ') + 1
\eneq
holds.
\end{lemma}
\begin{proof}
As we have seen $\tW ' = \tilde{L}_{\tW} \tilde{r}(k\pi) \tW_0$, with $\tW = \tilde{L}_{\tW} \tW_0$, which shows that \mbox{$\theta(\tW ') = \theta(\tW) + k\pi$}. Using this in the definition of the winding number leads to
\beq
\begin{split}
\theta(\tW ') + 2\pi N &< \theta(\tW) < \theta(\tW ') + 2\pi(N+1) \\
\Rightarrow \theta(\tW) + k\pi + 2\pi N &< \theta(\tW) < \theta(\tW) + k\pi + 2\pi(N+1).
\end{split}
\eeq 
From this it follows that
\beq
k+ 2N < 0 < k+2(N+1)
\eeq
which immediately leads to $-k = 2N + 1$.
\end{proof}
Having established the relation between our $k(\tW,\tW ')$ and the usual definition of the winding number $N(\tW,\tW ')$ we return to the calculation of $\frac{u_{\tW '}^\lambda(p)}{u_{\tW}^\lambda(p)}$. To determine the second factor in \eqref{uoveru} we need to compute the expression
\beq
\frac{u_0^\lambda(r(-k\pi)p)}{u_0^\lambda(p)} = \frac{u_0^\lambda(r(\pi)p)}{u_0^\lambda(p)},
\eeq
where $r(\pi)$ acts on $p$ simply as $r(\pi)(p_0,p_1,p_2) = (p_0,-p_1,-p_2)$. Inserting the definition of $u_0^\lambda$ and restricting the momentum $p$ to the forward mass shell $H_m^+$ leads via a straightforward calculation to
\beq
\left.\frac{u_0^1(r(\pi)p)}{u_0^1(p)}\right|_{H_m^+} = \left.\frac{f(-p_2) v(r(\pi)p)}{f(p_2)v(p)}\right|_{H_m^+} = \frac{f(-p_2)(m-ip_2)}{f(p_2)(m+ip_2)}.
\eeq
In order to guarantee the right commutation relations between our annihilation operators we therefore have to choose a function $f: \mathds{R} \rightarrow \mathds{C}$ satisfying
\bneq \label{condf}
\frac{f(-\kappa)(m-i\kappa)}{f(\kappa)(m+i\kappa)} = 1.
\eneq 
By studying the commutation relations between other creation and annihilation operators (e.g. between $a_{\tW}(p)$ and $b_{\tW '}(p')$) one realizes that $f$ also has to satisfy $\bar{f} = f^{-1}$, i.e. $|f| = 1$.
Using such an $f$ in the definition of $u_{\tW}^\lambda$ one then finally arrives at the desired relation
\bneq
\frac{u_{\tW '}^\lambda(p)}{u_{\tW}^\lambda(p)} = e^{-i\lambda \pi k(\tW,\tW ')}
\eneq
which leads to
\bneq
a_{\tW}(p) a_{\tW '}(p') = e^{-2\pi i\lambda k(\tW,\tW ')} a_{\tW '}(p') a_{\tW}(p).
\eneq
In exactly the same way one can calculate
\bneq
a_{\tW}(p) b_{\tW '}(p') = e^{2\pi i\lambda k(\tW,\tW ')} b_{\tW '}(p') a_{\tW}(p)
\eneq
and the mixed commutation relations
\bneq
\begin{split}
a_{\tW}(p) a_{\tW '}^*(p') &= e^{2\pi i\lambda k(\tW,\tW ')} a_{\tW '}^*(p') a_{\tW}(p) + \omega_p \delta(p-p') T_{\tW}(p) T_{\tW '}(p)^*, \\
a_{\tW}(p) b_{\tW '}^*(p') &= e^{-2\pi i\lambda k(\tW,\tW ')} b_{\tW '}^*(p') a_{\tW}(p).
\end{split}
\eneq
Again all the other commutation relations follow by charge conjugation and taking adjoints. 
\vspace{1em} \\
Up to now we have only taken into account covariance under the proper orthochronous group $\tilde{\mathcal{P}}_+^\uparrow$, but we also want our fields to have the correct transformation behavior under reflections at the edge of the wedge. Taking the standard wedge $\tW_0$ and the reflection $\tilde{j}$ at its edge we want the field to satisfy
\bneq \label{jtransformation}
J \Phi_{\tW_0}(f) J = \Phi_{\tilde{j}\tW_0}(\alpha_j(f)),
\eneq 
where again $\alpha_j(f)(x) = \overline{f(jx)}$ and $J$ has been defined in \eqref{jactionbeta}. According to \eqref{jaction} the reflection $\tilde{j}$ acts on $\tW_0$ as
\beq
\tilde{j}\cdot\tW_0 = (-W_0,\tilde{j}\cdot\tilde{e}_0) = \tilde{r}(-\pi)\tW_0.
\eeq
A straightforward calculation then shows that for equation \eqref{jtransformation} to hold we need to set $\beta = -2\pi\lambda$ in the definition \eqref{jactionbeta} and the intertwiners $u^\lambda_0$ have to satisfy $\overline{u^\lambda_0(-jp)} = u^\lambda_0(p)$. Because of $\tilde{j}\cdot\tW_0 = \tilde{r}(-\pi)\cdot\tW_0$ this leads to
\beq
\overline{u^\lambda_0(-jp)} =e^{-i\pi\lambda} u_{\tilde{j}\tW_0}(p).
\eeq
This equation can be fulfilled if and only if the function $f$, used in the definition of $u^\lambda_0$, satisfies $f(-\kappa) = \overline{f(\kappa)}$ in addition to the previous relation \eqref{condf}. These two conditions now lead to the solution
\bneq
f(\kappa) = \pm \frac{m-i\kappa}{\sqrt{m^2+\kappa^2}},
\eneq
which is unique up to a sign.
With this choice the deformed field $\Phi_{\tW_0}$ then satisfies \eqref{jtransformation} with the space-time reflection $J$ defined according to
\bneq
(J\Psi)_n^m(\up) = e^{-2\pi i\lambda q}\, \overline{\Psi_n^m(-j\up)}.
\eneq \vspace{1em}

Summing up our construction we have seen that we can deform the CCR-algebra of ``free'' creation and annihilation operators in such a way that the deformed operators satisfy anyonic commutation relations and the resulting field is covariant under a spin $\lambda$ representation of $\tilde{\mathcal{P}}_+$. This deformation was defined by simply using multiplication operators $(T_{\tW}(p)\Psi)_n^m(\up) = A_{\tW}^{n,m}(p;\up) \Psi_n^m(\up)$ which were chosen according to
\beq
\begin{split}
&A_{\tW}^{n,m}(p;\up) = u_{\tW}^\lambda(p)^{q+1}\prod_{i=1}^{n}u_{\tW}^\lambda(p_i)R(Qp\cdot p_i)\prod_{j=n+1}^{n+m}\overline{u_{\tW}^\lambda(p_j)} R(Qp\cdot p_j), \\
&R \ \ ...\ \text{ arbitrary deformation function in the sense of \cite{Lechner12}},\\
&u_ {\tW}^\lambda(p) = e^{-i\lambda\Omega(\tilde{L}_{\tW},p)} u_0^\lambda(\tilde{L}_{\tW}^{-1}p), \\
&u_0^\lambda(p) = \bigg(\frac{m-ip_2}{\sqrt{m^2+p_2^2}}\bigg)^\lambda\cdot v(p)^\lambda \ \ , \hspace{1.5em} v(p) := \frac{p_0+m-p_1+ip_2}{p_0+m-p_1-ip_2}.
\end{split}
\eeq \vspace{1em} \\
Having defined the deformation we can now state our main result, namely that the deformed field satisfies anyonic commutation relations.
\begin{theorem}
Consider paths of wedges $\tW$, $\tW '$ and test functions $f, g$ such that 
\bneq \label{suppprop}
{\operatorname{supp}(f) + W} \subset (\operatorname{supp}(g) + W')'.
\eneq
Then the deformed fields $\Phi_{\tW}(f)$ and $\Phi_{\tW '}(g)$ satisfy the commutation relations
\bneq
\begin{split}
\Phi_{\tW}(f) \Phi_{\tW '}(g) &= e^{-2\pi i\lambda k(\tW,\tW ')}\Phi_{\tW '}(g) \Phi_{\tW}(f), \\
\Phi_{\tW}(f) \Phi^*_{\tW '}(g) &= e^{2\pi i\lambda k(\tW,\tW ')}\Phi^*_{\tW '}(g) \Phi_{\tW}(f),
\end{split}
\eneq
\end{theorem}
\begin{proof}
Using the above relations between the deformed creation and annihilation operators one immediately sees that the deformed field $\Phi_{\tW}(f) = a^*_{\tW}(f^+) + b_{\tW}(\bar{f}^+)$ with itself satisfies the commutation relation
\bneq
\Phi_{\tW}(f) \Phi_{\tW '}(g) = e^{-2\pi i\lambda k(\tW,\tW ')}\Phi_{\tW '}(g) \Phi_{\tW}(f),
\eneq
for arbitrary test functions $f$ and $g$. \\
For the mixed commutation relations between $\Phi_{\tW}$ and $\Phi^*_{\tW '}$ we get
\bneq \label{fieldcommrel3d}
\begin{split}
&\Phi_{\tW}(f) \Phi^*_{\tW '}(g) - e^{2\pi i \lambda k(\tW,\tW')} \Phi^*_{\tW '}(g) \Phi_{\tW}(f) = \\
\int d\mu(p) &\left(f^-(p) g^+(p)\ T^c_{\tW}(p) T^c_{\tW '}(p)^* -  f^+(p) g^-(p)\ e^{2\pi i\lambda k(\tW,\tW ')} T_{\tW '}(p) T_{\tW}(p)^*\right).
\end{split}
\eneq
To determine whether the right hand side vanishes for spacelike separated testfunctions we have to calculate how the operators $T^c_{\tW}(p) T^c_{\tW '}(p)^*$ and $T_{\tW '}(p) T_{\tW}(p)^*$ act on an arbitrary vector $\Psi_n^m$. This leads to
\beq
\begin{split}
(T^c_{\tW}(p) T^c_{\tW '}(p)^* \Psi)_n^m(\up) = u_{\tW}^\lambda(p)^{-q+1} \overline{u_{\tW '}^\lambda(p)^{-q+1}}&\ \prod_{i=1}^n \overline{u_{\tW}^\lambda(p_i)} u_{\tW '}^\lambda(p_i) R(Qp\cdot p_i)^2 \\ \cdot 
&\prod_{j=n+1}^{n+m} u_{\tW}^\lambda(p_j) \overline{u_{\tW '}^\lambda(p_j)}R(Qp\cdot p_j)^2\, , \\
(T_{\tW '}(p) T_{\tW}(p)^* \Psi)_n^m(\up)\, = \, \overline{u_{\tW}^\lambda(p)^{q+1}}\, u_{\tW '}^\lambda(p)^{q+1}&\ \prod_{i=1}^n \overline{u_{\tW}^\lambda(p_i)} u_{\tW '}^\lambda(p_i) \overline{R(Qp\cdot p_i)}^2 \\ \cdot 
&\prod_{j=n+1}^{n+m} u_{\tW}^\lambda(p_j) \overline{u_{\tW '}^\lambda(p_j)}\overline{R(Qp\cdot p_j)}^2 \, .
\end{split}
\eeq
One can see that the factors containing $u^\lambda(p_i)$ and $u^\lambda(p_j)$ are the same in both equations, so they are not causing any trouble. The deformation function $R$ has been chosen in such a way that it has the right analytic properties, as in the two-dimensional case (cf. also \cite{GrossLech}). The only nontrivial factors left which could still cause problems are $u_{\tW}^\lambda(p)^{-q+1} \overline{u_{\tW '}^\lambda(p)^{-q+1}}$ and $\overline{u_{\tW}^\lambda(p)^{q+1}}\, u_{\tW '}^\lambda(p)^{q+1}$. But a straightforward calculation using the definition of $u_{\tW}^\lambda$ leads to
\beq
u_{\tW}^\lambda(p)^{-q+1} \overline{u_{\tW '}^\lambda(p)^{-q+1}} = e^{-i\pi\lambda k(\tW,\tW ')(-q+1)} = e^{2\pi i\lambda k(\tW,\tW ')}\overline{u_{\tW}^\lambda(p)^{q+1}}\, u_{\tW '}^\lambda(p)^{q+1} .
\eeq
Inserting this into the commutation relation \eqref{fieldcommrel3d} one arrives at
\bneq \label{commrelproof}
\begin{split}
&\hspace{-1cm}\left(\Phi_{\tW}(f) \Phi^*_{\tW '}(g)\Psi\right)_n^m(\up) - \left(e^{2\pi i \lambda k(\tW,\tW')} \Phi^*_{\tW '}(g) \Phi_{\tW}(f)\Psi\right)_n^m(\up) = \\
&\int d\mu(p) \Big[f^-(p) g^+(p) \prod_{i=1}^{n+m} R(Qp\cdot p_i)^2 - f^+(p) g^-(p) \prod_{i=1}^{n+m} \overline{R(Qp\cdot p_i)}^2 \Big]\\
&\cdot e^{-i\pi\lambda k(\tW,\tW ')} \prod_{i=1}^n \overline{u_{\tW}^\lambda(p_i)} u_{\tW '}^\lambda(p_i) \prod_{j=n+1}^{n+m} u_{\tW}^\lambda(p_i) \overline{u_{\tW '}^\lambda(p_i)} \ \ \Psi_n^m(\up).
\end{split}
\eneq
We now just have to show that the expression in the second line vanishes for all $p_i$ if the test functions $f$ and $g$ have the right support properties \eqref{suppprop}. Due to the covariance of the field operators we have to consider this expression only for the standard wedge $W_0$, i.e. $Q=Q_0$ and $f, g$ localized in $W_0, W_0'$ respectively. Using ideas from the proof of Proposition 3.4 in \cite{GrossLech07} one introduces new coordinates on the mass shell such that
\bneq
p = p(\theta,p_2) = \begin{pmatrix} m_\perp \cosh{\theta} \\ m_\perp \sinh{\theta} \\ p_2 \end{pmatrix} \ , \hspace{1cm} \int d\mu(p) \varphi(p)= \int \frac{1}{2} d\theta dp_2\ \varphi(p(\theta,p_2)),
\eneq
where $m_\perp = \sqrt{m^2+p_2^2}$. Because $f$ and $g$ are assumed to have compact support their Fourier transforms $\tilde{f}$ and $\tilde{g}$ are entire analytic functions. Moreover $f^-(\theta+i\lambda,p_2)$ and $g^+(\theta+i\lambda,p_2)$ are bounded on the strip $0\leq\lambda\leq\pi$ (see again \cite{GrossLech07}) and the boundary values are related by
\beq
f^-(\theta+i\pi,p_2) = f^+(\theta,-p_2) \ , \hspace{1em} g^+(\theta+i\pi,p_2) = g^-(\theta,-p_2),
\eeq
where we used the obvious notation $f^\pm(\theta,p_2) = \tilde{f}(\pm p(\theta))$.
Furthermore, because $\kappa \geq 0$ and all the momenta are on the mass shell, it follows that
\beq
\operatorname{Im}(Q p(\theta+i\lambda)\cdot p_k) = \kappa m_\perp \sin{\lambda} \binom{\cosh{\theta}}{\sinh{\theta}}\cdot \binom{p_k^0}{p_k^1} \geq 0 \ , \hspace{1em} \text{for } 0 \leq \lambda \leq \pi .
\eeq
Therefore the functions $z \mapsto R(Q p(z)\cdot p_k))$ are analytic on the strip $S(0,\pi)$ and bounded on its closure. This allows us to shift the $\theta$ integration in the second line of \eqref{commrelproof} from $\mathds{R}$ to $\mathds{R} + i\pi$, which shows that the whole expression vanishes (See also \cite{Alazzawi12} for a more detailed treatment of these concepts). \\
We have thus shown that the deformed field satisfies the anyonic commutation relations
\beq
\begin{split}
\Phi_{\tW}(f) \Phi_{\tW '}(g) &= e^{-2\pi i\lambda k(\tW,\tW ')}\Phi_{\tW '}(g) \Phi_{\tW}(f), \\
\Phi_{\tW}(f) \Phi^*_{\tW '}(g) &= e^{2\pi i\lambda k(\tW,\tW ')}\Phi^*_{\tW '}(g) \Phi_{\tW}(f),
\end{split}
\eeq
if the localization regions $(\operatorname{supp}(f) + W)''$ and $(\operatorname{supp}(g) + W')''$ are spacelike separated.
\end{proof}
\vspace{1em}
\sloppy{Summing up our results we have constructed field operators $\Phi_{\tW}(f)$ on the Hilbert space $\mathcal{F}_s(L^2(\mathds{R}^3,d\mu))\otimes\mathcal{F}_s(L^2(\mathds{R}^3,d\mu))$} for every path of wedge $\tW$ and test function $f\in\mathscr{S}(\mathds{R}^3)$. This family satisfies the following properties:
\begin{enumerate}[i)]
\item The fields are \emph{polarization-free generators}, i.e. 
\beq
\Phi_{\tW}(f)\Omega \in \mathcal{H}_1 .
\eeq
\item \emph{Covariance} under the representation \eqref{representation} of $\tilde{\mathcal{P}}_+^\uparrow$ holds, i.e.
\beq
U(a,\tL)\Phi_{\tW}(f)U(a,\tL)^{-1} = \Phi_{\tL\tW}(\alpha_{(a,\Lambda)}(f)) .
\eeq
\item Under the representation $J$ of the reflection at the $x_2$-axis the field transforms according to
\beq
J \Phi_{\tW_0}(f) J = \Phi_{\tilde{j}\tW_0}(\alpha_j(f)).
\eeq
\item The fields are \emph{localized in wedge regions} and satisfy \emph{anyonic commutation relations}, depending on the relative winding number of $\tW$ and $\tW '$, i.e.
\beq
\Phi_{\tW}(f)\Phi^\sharp_{\tW '}(g) = e^{\mp 2\pi i\lambda k(\tW,\tW ')} \Phi^\sharp_{\tW '}(g) \Phi_{\tW}(f),
\eeq
if $\operatorname{supp}(f)+W \subset (\operatorname{supp}(g)+W')'$.
\end{enumerate}
Furthermore the \emph{Reeh-Schlieder} property holds for wedges, but what's more important is that it does \emph{not} hold for regions smaller than a wedge, e.g. for spacelike cones. This follows from the recent work \cite{BrosMund} by Bros and Mund, where they show (using results from \cite{BBS}) that there can be no polarization-free generators for Anyons if the Reeh-Schlieder property holds for spacelike cones. \\ \vspace{1em}

\subsection{Scattering States} \label{sec:scattering}
We can now also define two-particle scattering states and their S-matrix following the approach for wedge-localized operators in \cite{BBS}. For this purpose we choose $f$ and $g$ in such a way that $\tilde{f}, \tilde{g}$ have compact support and $\operatorname{supp}\tilde{f}, \operatorname{supp}\tilde{g}$ intersect the upper but not the lower mass shell. Furthermore, we use the notation $p = (p^0,\mathbf{p})$, $\omega_p = \sqrt{\mathbf{p}^2+m^2}$ and define the velocity support $\Gamma(f) := \{(1,\mathbf{p}/\omega_p) : p \in \operatorname{supp}\tilde{f}\}$ of $f$ and its time evolution $f_t(x) = \int dp \tilde{f}(p) e^{i(p^0-\omega_p)t} e^{-ip\cdot x}$. It is well known in scattering theory that for large times $t$ the fields $\Phi_{\tW}(f_t)$ and $\Phi_{\tW '}(g_t)$ are essentially localized in $W + t\Gamma(f)$ and $W'+t\Gamma(g)$ respectively, because asymptotically the support of $f_t$ is contained in $t\Gamma(f)$. Now if the velocity supports of $f, g$ are such that $\Gamma(f) - \Gamma(g) \subset W$ these localization regions are spacelike separated for positive times $t$. Therefore we can define outgoing two-particle scattering states as the limit
\bneq
\lim_{t\to\infty} \Phi_{\tW}(f_t)\Phi_{\tW '}(g_t) =: (f^+,\tW)_+ \times^{out} (g^+,\tW ')_+
\eneq
and by using also $\Phi^*_{\tW}(f)$ we could similarly construct scattering states containing anti-particles. \\
For the incoming scattering states we have to exchange $\tW$ and $\tW '$ because for $t<0$ the localization regions $W+t\Gamma(g)$ and $W'+t\Gamma(f)$ are spacelike separated. This leads to the definition
\bneq
\lim_{t\to-\infty} \Phi_{\tW '}(f_t)\Phi_{\tW}(g_t) =: (f^+,\tW ')_+ \times^{in} (g^+,\tW)_+ .
\eneq
To compute these limits we use that the supports of $f$, $g$ do not intersect the lower mass shell and that the time dependence of $f_t$ is trivial on the upper mass shell. Thus the scattering states simplify to
\bneq \label{scatstates}
\begin{split}
(f^+,\tW)_+ \times^{out} (g^+,\tW ')_+ &= a^*_{\tW}(f^+) a^*_{\tW '}(g^+)\Omega, \\
 (f^+,\tW ')_+ \times^{in} (g^+,\tW)_+ &= a^*_{\tW '}(f^+) a^*_{\tW}(g^+)\Omega.
\end{split}
\eneq
Of course these vectors inherit the nontrivial commutation relations from the fields which create them, leading e.g. to
\bneq
(f^+,\tW)_+ \times^{out} (g^+,\tW ')_+ = e^{-2\pi i\lambda k(\tW,\tW ')}\  (g^+,\tW ')_+ \times^{out} (f^+,\tW )_+.
\eneq
From \eqref{scatstates} the explicit form of the scattering states can be computed, namely
\bneq
\begin{split}
\left((f^+,\tW)_+ \times^{out} (g^+,\tW ')_+\right)(p_1,p_2) &= \frac{1}{\sqrt{2}} \Big( \mathfrak{R}(p_1,p_2) f^+(p_1) g^+(p_2)\ +\ (p_1 \leftrightarrow p_2) \Big), \\
\left((f^+,\tW ')_+ \times^{in} (g^+,\tW)_+\right)(p_1,p_2) &= \frac{1}{\sqrt{2}} \Big( \mathfrak{R}'(p_1,p_2) f^+(p_1) g^+(p_2)\ +\ (p_1 \leftrightarrow p_2) \Big),
\end{split}
\eneq
where $\mathfrak{R}$ and $\mathfrak{R}'$ are defined according to
\beq
\begin{split}
\mathfrak{R}(p_1,p_2) &:= \overline{u_{\tW}^\lambda(p_1)}^2 \overline{u_{\tW}^\lambda(p_2)}\, \overline{u_{\tW '}^\lambda(p_2)}\, R(Qp_1\cdot p_2), \\
\mathfrak{R}'(p_1,p_2) &:= \overline{u_{\tW '}^\lambda(p_1)}^2 \overline{u_{\tW '}^\lambda(p_2)}\, \overline{u_{\tW}^\lambda(p_2)}\, \overline{R(Qp_1\cdot p_2)}.
\end{split}
\eeq
Now taking testfunctions $f, g, h, k \in \mathscr{S}(\mathds{R}^3)$ such that $\Gamma(f) - \Gamma(g) \subset W$ and $\Gamma(h) - \Gamma(k) \subset W$ we can calculate the two-particle S-matrix through
\bneq
\begin{split}
\langle (f^+,\tW)_+ &\times^{out} (g^+,\tW ')_+, (h^+,\tW ')_+ \times^{in} (k^+,\tW)_+\rangle \\
&= \int d\mu(p_1)d\mu(p_2)\left( e^{2\pi i\lambda k(\tW,\tW ')} R(p_1Qp_2)^2\right) \overline{f^+(p_1)}\overline{g^+(p_2)}h^+(p_1) k^+(p_2).
\end{split}
\eneq
So we can see that the momentum dependence of the S-matrix is again determined by $R^2$ which in addition gets multiplied by a phase factor depending only on the relative winding number of the localization regions. Note that we would get the same result for anti-particle scattering because we used the same deformation function on the particle- and the anti-particle space. Of course we could also have written $R^+$ and $R^-$ as in the two-dimensional case, but it is presently unclear to what extend these two functions could differ in $d>2$.
\vspace{1em}

\section{Conclusion and Open Questions}
The method of deforming a free hermitian scalar field to obtain new wedge-localized models with non-trivial S-matrix, described in detail in \cite{Lechner12}, has been generalized to a \emph{charged} field on two- and three-dimensional Minkowski space. We have seen that working on a charged Hilbert space enables us to change the statistics of the deformed fields and, at least in two-dimensions, one can use a larger class of deformation functions in the definition of the deformation. The fields are localized in so-called paths of wedges, have non-trivial commutation relations depending on the winding number of their localization regions and they are covariant with respect to a representation of $\tilde{\mathcal{L}}_+^\uparrow$ with spin $\lambda\in\mathds{R}$. \\
The weakened localization in wedges instead of double cones or spacelike cones still allows to define two-particle scattering states and the wedge-algebras generate dense sets in $\mathcal{H}$ from the vacuum. However, nothing is known about algebras for smaller spacetime regions, except that they cannot satisfy the Reeh-Schlieder property due to the recent work by Bros and Mund \cite{BrosMund}. \\
Another open question is if the deformation functions defined in \eqref{generaldeformation} are actually the most general which are possible in this case. We currently also do not know if these additional admissible functions can be generalized to the higher dimensional case. \\
In \cite{Lechner12} the deformation was originally defined on the testfunction algebra (the so-called Borchers-Uhlmann algebra) and a representation on Fock space was then introduced via the GNS-construction. In the present work we defined the deformation on Fock space from the outset and it would be interesting to know if one can understand it as the GNS representation of a deformation of the underlying Borchers-Uhlmann algebra. \\
In addition one would like to generalize the deformation of a charged scalar field to a situation with a multi-component field where different particle species and charges are present. In this case the commutation relations might be governed not only by phase factors, but by more general matrices yielding a non-Abelian representation of the braid group.

\vspace{1em}
\subsection*{Acknowledgments}
I would like to thank the Vienna deformation group - S. Alazzawi, C. Köhler, M. Könenberg and J. Schlemmer - for helpful discussions. I am furthermore indebted to Gandalf Lechner and my advisor Jakob Yngvason for providing constant support. I also want to thank Jens Mund and Bert Schroer for comments on a first draft of this paper. \\
This work was supported by the FWF-project P22929-N16 ``Deformations of Quantum
Field Theories''. 
\vspace{2em}

%\bibliographystyle{amsplain}
%\bibliography{bibliography}

\end{document}